\documentclass[conference]{IEEEtran}
\usepackage{courier}
\usepackage{color,url,xspace}
\usepackage{amssymb,amsmath}
\usepackage{graphicx}
\usepackage{colordvi}
\usepackage{epsfig}
\usepackage{endnotes}
\usepackage{verbatim}
\usepackage{subfigure}
\usepackage{textcomp}
\usepackage {subfig}
\usepackage{setspace}
\usepackage{url}
\usepackage{amsthm}
\begin{document}
\title{Quality-Aware Coding and Relaying for 60\,GHz Real-Time Wireless Video Broadcasting}
\author{\IEEEauthorblockN{Joongheon Kim$^\dag$, \textit{Member, IEEE}, Yafei Tian$^\natural$, \textit{Member, IEEE}, Stefan Mangold$^\S$, \textit{Member, IEEE}, and\\Andreas F. Molisch$^{\ddag}$, \textit{Fellow, IEEE}}
\IEEEauthorblockA{$^{\dag\ddag}$Communication Sciences Institute, University of Southern California, Los Angeles, CA 90089, USA\\
$^\natural$School of Electronics and Information Engineering, Beihang University, Beijing 100191, China\\
$^\S$Disney Research, 8092 Zurich, Switzerland\\
Emails: $^\dag$joongheon.kim@usc.edu, $^\natural$ytian@buaa.edu.cn, $^\S$stefan@disneyresearch.com, $^{\ddag}$molisch@usc.edu}
}

\maketitle

\begin{abstract}
Wireless streaming of high-definition video is a promising application for 60\,GHz links, since multi-Gigabit/s data rates are possible.
In particular we consider a sports stadium broadcasting system where video signals from multiple cameras are transmitted to a central location.
Due to the high pathloss of 60\,GHz radiation over the large distances encountered in this setting, the use of relays is required.
This paper designs a quality-aware coding and relaying algorithm for maximization of the overall video quality.
We consider the setting that the source can split its data stream into parallel streams, which can be transmitted via different relays to the destination.
For this, we derive the related formulation and re-formulate it as convex programming, which can guarantee optimal solutions.
\end{abstract}

\IEEEpeerreviewmaketitle

\section{Introduction}\label{sec:intro}
Wireless video streaming in the millimeter-wave range has received a lot of attention in both the academic and industrial communities. In particular the 60\,GHz frequency range is of great interest: around 7\,GHz bandwidth (58-65\,GHz) has been made available, which enables multi-Gbit/s high-definition video streaming in an uncompressed, or less compressed, manner.
Therefore, two industry consortia, i.e.,
	WirelessHD and Wireless Gigabit Alliance (WiGig), have developed related specifications; there are also two activities within the IEEE, namely IEEE 802.15.3c~\cite{153c} and IEEE 802.11ad~\cite{11ad}.

In this paper, we design and analyze such a 60\,GHz video transmission system for outdoor applications, in particular in a sports stadium.
In this system, there are multiple wireless video cameras in a stadium for high-quality real-time broadcasting, all send their signals to a broadcasting center.
To transmit uncompressed HD video streams in real-time, a data rate of around $1.5$ Gbit/s is required~\cite{kim11pimrc}.
Since the distance between wireless cameras and a broadcasting center is on the order of several hundred meters, the high pathloss at 60\,GHz is one of key challenges that limits communication ranges.
One promising approach to deal with this problem is using \textit{relays} for extending the coverage~\cite{kim11pimrc}.
Additionally, we take the complexity of the antennas into account. In order to compensate for the high pathloss, as well as to reduce interference, high-gain antennas need to be employed. We also consider the situation where the antenna at the camera (video source) can form multiple beams, so that it can split its data stream into multiple streams and send them to the destination via parallel links.
By introducing multiple beams in each relay, our framework operates even though the number of relays is smaller than the number of sources.
In this case, appropriate compression and routing of multiple streams via the same relay can be used.
Relaying for sum rate maximization has been analyzed in many papers. However, for video streaming, we are more interested in {\em video quality}. For this, the proposed quality-aware formulation selects the relays and decides the coding rates for every single video stream. With this formulation, optimal solutions are obtained by convex optimization techniques.

Thus, the contribution of the proposed scheme is achieving \textit{joint rate and relay selection} with \textit{video quality} consideration and an \textit{interference-free} operation. This combination of special features makes it different from other schemes.

\begin{figure*}[t!]
	\centering
	\subfigure[Wireless Video Camera (Source)]{
		\includegraphics[scale=0.23]{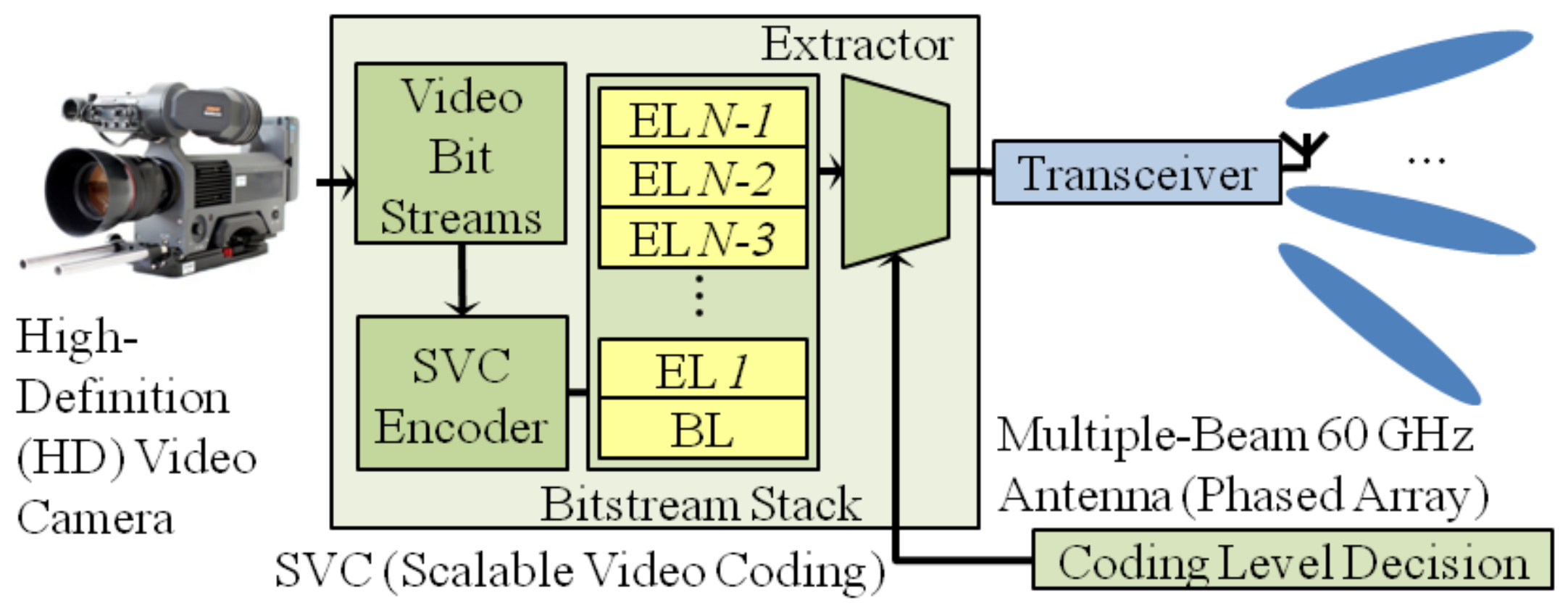}
		\label{fig:vcs}
	}
	\subfigure[Relay]{
		\includegraphics[scale=0.23]{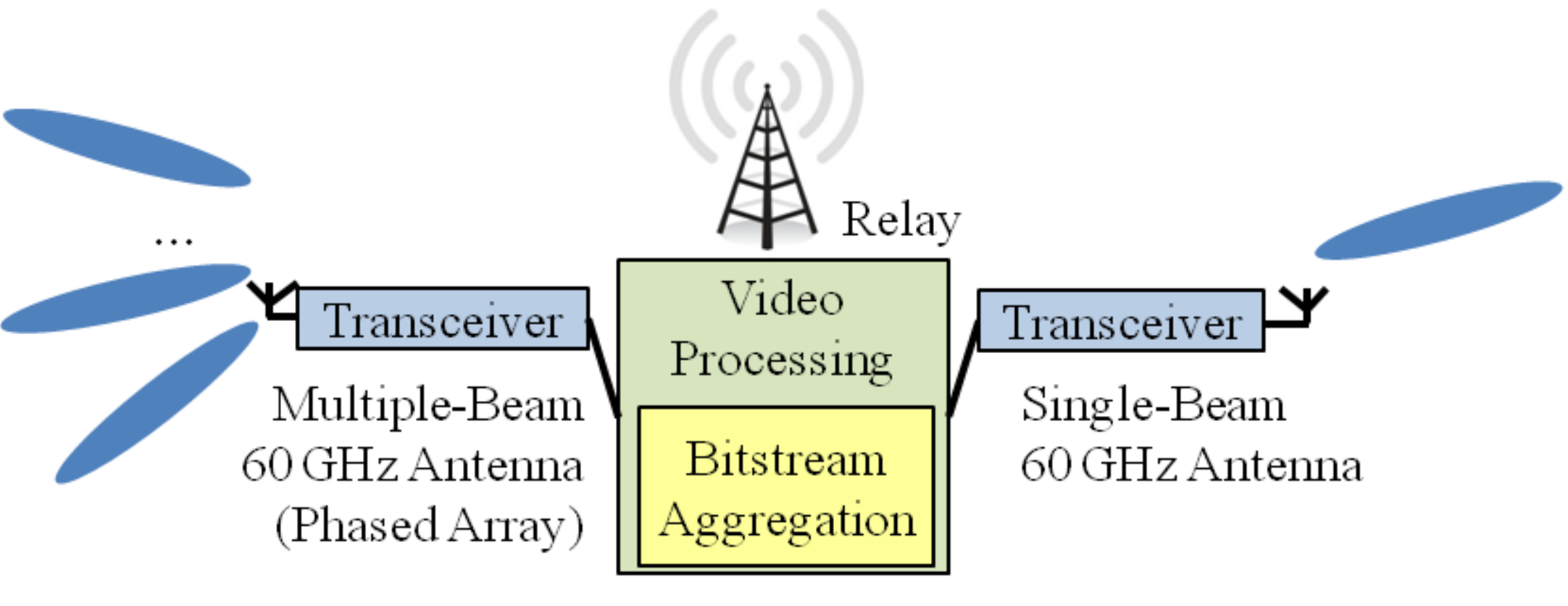}
		\label{fig:relay}
	}
	\subfigure[Broadcasting Center]{
		\includegraphics[scale=0.23]{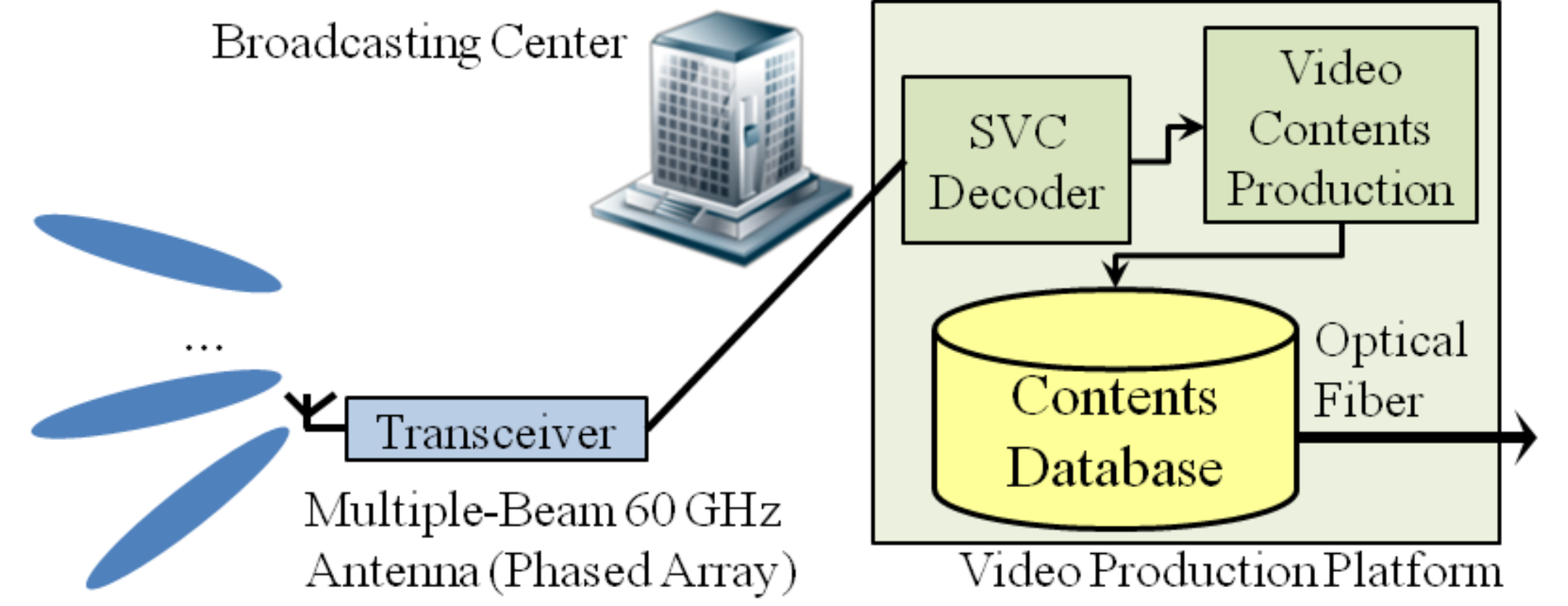}
		\label{fig:bc}
	}
	\caption{System Components (Camera (a), Relay (b), Broadcasting Center (c))}
	\label{fig:sysmodule}
\end{figure*}

\begin{figure}[t!]
	\begin{center}
		\includegraphics[scale=0.33]{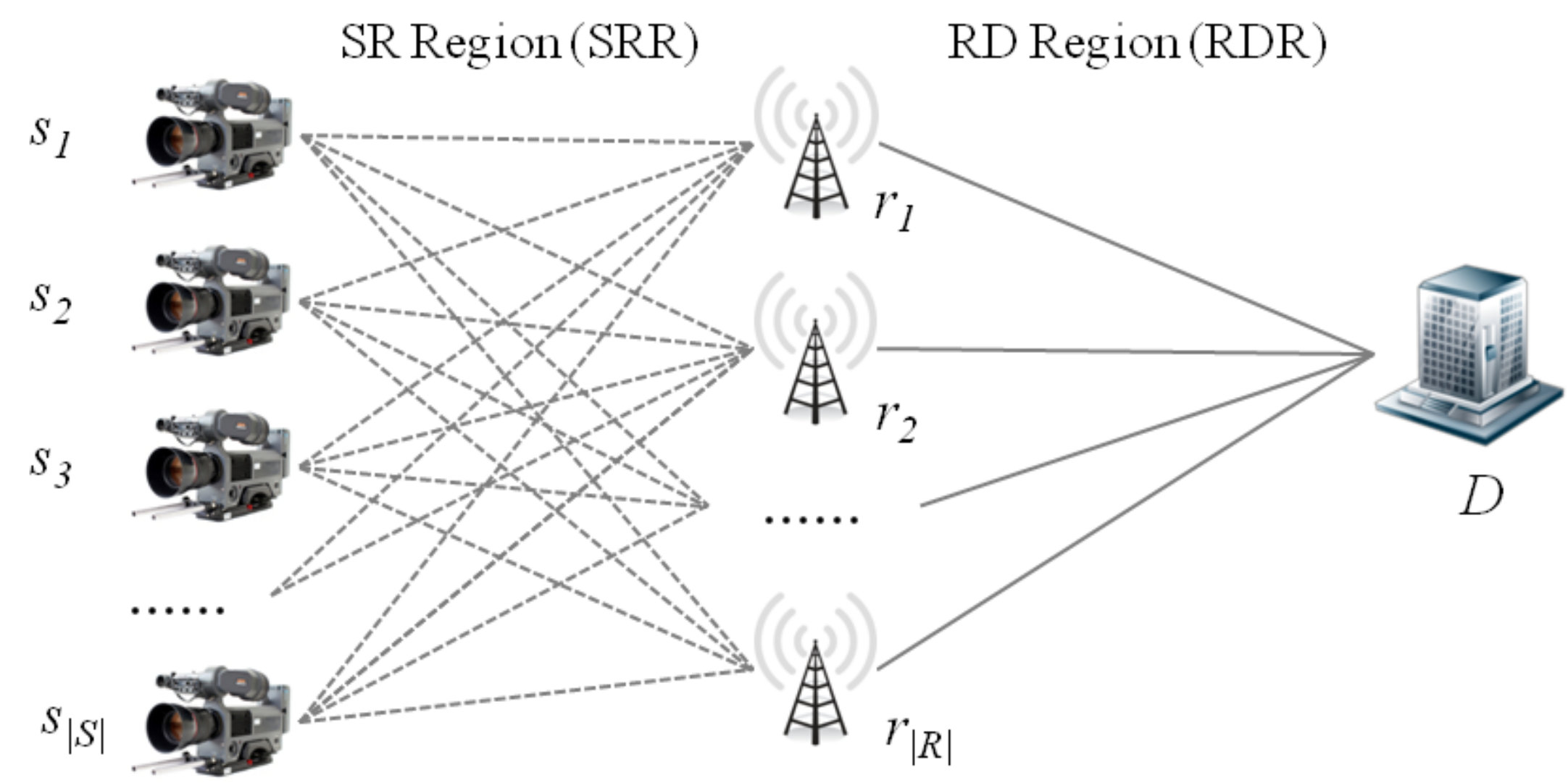}
	\end{center}
	\caption{Overall Architecture}
	\label{fig:sysmodel}
\end{figure}	

\section{Related Work}\label{sec:related}

There are two salient factors in our broadcasting setup: (i) stream splitting via the multiple-beam antennas, and (ii) rate control for video quality maximization. In the following, we outline why these aspects make the setup different from other scenarios that have been treated in the literature.
There is, of course, a huge number of papers (too numerous to reference here) dealing with the topic of routing (relay selection) and sum rate maximization in multi-node networks with multi-beam relays~\cite{yilmaz10icc} and with multi-beam sources and relays~\cite{liu12jsac}. However these papers do not consider the control of the video coding rate (compression), and are thus not directly applicable to our scenario.
For video networks, example publications include~\cite{videorelay1, videorelay2, videorelay3}. The scheme in~\cite{videorelay1} is for video streaming over IEEE 802.11 networks. The proposed scheme is efficient for the multi-hop networks it investigates; however, it does not consider the video stream splitting via the multi-beam antennas and route selection.
Ref.~\cite{videorelay2} considers video streaming in multi-hop networks. It considers networks similar to ours (when specialized to the two-hop case), but again does not investigate multi-beam antennas and the splitting of the data streams. The formulation in~\cite{videorelay3} considers path selection for video streaming in MANET. It concentrates on the consideration of interference, a factor that does not play a role in our 60 GHz channel, where the high directionality of the links prevents inter-stream interference. None of these papers consider the control of the coding rate (compression).
In previous research on video streaming, schemes usually considered multipath transmission to combat the limited bandwidth~\cite{videorelay2}\cite{videorelay3}. Also, some of the research considered retransmission of frames and tried to reduce transmission time~\cite{videorelay1}. However, thanks to the extremely large bandwidth at 60 GHz, these factors are not longer critical in our system.
The representative work which considers both rate control and routing appeared in~\cite{jsac12thou}: However, the relays cannot aggregate streams, which is required when the number of relays is smaller than the number of flows. In addition, the proposed framework does not consider the properties of video. In our previous work~\cite{kim11pimrc}, we considered the properties of 60 GHz channel, rate control, and video quality, but we restricted ourselves to the cases that the number of relays exceeds the number of sources (i.e., no consideration for multi-beam antennas) and the numbers of sources and destinations are identical.
We finally note that wireless video for sports stadiums~\cite{disney_wowmom}; however the fundamental setup differs from ours in that~\cite{disney_wowmom} considers content distribution to wireless devices of the audience in the stadium, while our system is for real-time streaming to a broadcasting center in the stadium.

\section{A Reference System Architecture}\label{sec:refsysmodel}

\subsection{Link Budget Analysis}\label{sec:refsysmodel-lba}
Shannon's equation for the capacity is used for the data rate:
\begin{equation}
C = B\cdot\log_{2}\left(1+\text{SNR}\right)\label{eq:shannon}
\end{equation}
where $\text{SNR}$ is equal to $P_{\text{signal}}/P_{\text{noise}}$ on a linear scale, $P_{\text{signal}}$ and $P_{\text{noise}}$ stand for the signal power and noise power, and $B$ stands for bandwidth ($2.16$\,GHz in WiGig~\cite{11ad})~\cite{molisch11book}.
The signal power expressed in dB, $P_{\text{signal, dB}}$, is obtained as:
\begin{equation}
P_{\text{signal,dB}} = E + G_{r} - W -O(d) + F(d)
\end{equation}
where $E$ denotes the EIRP (equivalent isotropically radiated power), which is limited to $40$\,dBm in the USA and $57$\,dBm in Europe.
$G_{r}$ means the receiver antenna gain and is set to $40$\,dB, which corresponds to high-gain 60\,GHz scalar horn antennas~\cite{comotech}, which we propose to achieve long range.
Shadowing can be either temporally variant (due to people walking close to LOS), or time-invariant (due to objects that are (partly) shadowing off the LOS. While the shadowing variances envisioned for our deployment scenarios are on the order of a few dB, we use a $10$\,dB shadowing margin to provide high link reliability.
$F(d)$ is the mean pathloss, which depends on the distance $d$ between transmitter and receiver
\begin{equation}
F(d) = 10\log_{10}\left\{\lambda/(4\pi d)\right\}^{n}
\end{equation}
where the pathloss coefficient $n$ is set to $2.5$~\cite{60ghz} and the wavelength ($\lambda$) is
$5$\,millimeter at $60$\,GHz.
$O(d)$ denotes the oxygen attenuation, which can be computed as $O(d) = \frac{15}{1000}d$ when $d > 200$\,m. Otherwise, it is ignored~\cite{60ghz}.
The noise power in dB, $P_{\text{noise,dB}}$ is computed as:
\begin{equation}
P_{\text{noise,dB}} = 10\log_{10}\left(k_{B}T_{e}\cdot B\right) + F_{N}
\end{equation}
where $k_{B}T_{e}$ stands for the noise power spectral density ($-174$\,dBm/Hz) and $F_{N}$ is the noise figure of the receiver ($6$\,dB).
By combining the above equations, approximately $200 - 300$\,m is the maximum distance for obtaining $1.5$ GBit/s data rate, i.e., successful uncompressed video transmission.


\subsection{60\,GHz Outdoor Broadcasting Systems}\label{sec:refsysmodel-model}

As shown in Sec.~\ref{sec:refsysmodel-lba}, the assistance of relays is required if the distance between wireless cameras and a broadcasting center is more than $200 - 300$\,m.
Furthermore, the size of a sports stadium (from wireless cameras to a broadcasting center) is generally not more than $500$\,m. Thus, we restrict the number of relays to one.
In our 60\,GHz broadcasting system, three components are existing, i.e., wireless video cameras, relays, and a broadcasting center.
As presented in Fig.~\ref{fig:vcs}, the proposed wireless video cameras have scalable video coding (SVC) functionalities that reproduce the recorded video signals as layered SVC-coded bitstreams.
If the achievable rate of a 60\,GHz link is sufficient for uncompressed video streaming (i.e., more than $1.5$\,Gbit/s), all layers can be transmitted.
Otherwise, the optimal coding level decision module has to determine the number of layers.
Each wireless video camera has multiple-beam antennas.
Therefore, each antenna can form $N$ independent beams, so that the multiple streams created by SVC-encoding are divided into $N$ parts and each part is assigned to a beam to be concurrently transmitted.
We furthermore assume that the relays have multiple-beam antennas for reception, see
Fig.~\ref{fig:relay}. The relays aggregate the received signals and transmit them towards a broadcasting center. As presented in Fig.~\ref{fig:bc}, the proposed broadcasting center has multiple antennas which are facing the relays. We emphasize that due to the narrow beamwidth ($1.5^{\circ}$-$10^{\circ}$~\cite{comotech}) of the antennas, multiple streams arriving at the broadcasting center or relays do not interfere with each other.

\subsection{Objective}\label{sec:objective}
For this given system, our objective is the maximization of the delivered total video quality.
As shown in~\cite{kim12tbc}, the quality of video is related to the data rate in a nonlinear fashion as a sublinearly, but monotonically, increasing form.
Following is one example of a quality function:
\begin{equation}
f_{q}(a) = \frac{1}{\log_{\beta}(a_{\text{max}}+1)} \log_{\beta}(a+1)
\end{equation}
$\beta$ is a base ($1<\beta$), $a_{\text{max}}$ is a desired data rate for uncompressed video streaming, and $a$ is a given data rate.


\section{Mathematical Optimization Formulation}\label{sec:jscr}

Fig.~\ref{fig:sysmodel} shows the reference model with a set of sources $\mathcal{S}$, a set of relays $\mathcal{R}$, and a single destination $D$.
In the relay-destination region (RDR) of Fig.~\ref{fig:sysmodel}, all relays are connected to $D$.
Then the maximum achievable rates of all relay and destination pairs can be computed (i.e., $a_{r_{1}\rightarrow D}^{\text{RDR}},\cdots, a_{r_{|\mathcal{R}|}\rightarrow D}^{\text{RDR}}$). Our assumption is that $D$ can form a sufficient number of independent beams so that it has no limitations concerning the number of relays.
Thus, we wish to find optimal combinations between sources and relays in source-relay region (SRR) for the settings that both sources and relays can form multiple beams.
Then, our formulation for the maximization of delivered total video quality is as follows:
\begin{equation}
\max \sum_{j=1}^{|\mathcal{R}|}\sum_{i=1}^{|\mathcal{S}|}f_{q}\left(\frac{1}{2}a_{s_{i}\rightarrow r_{j}}^{\text{SRR}}\right)x_{s_{i}\rightarrow r_{j}}^{\text{SRR}}\label{eq:sss-objfunc}
\end{equation}
subject to
\begin{eqnarray}
\sum_{i=1}^{|\mathcal{S}|}a_{s_{i}\rightarrow r_{j}}^{\text{SRR}}x_{s_{i}\rightarrow r_{j}}^{\text{SRR}} &\leq& \mathcal{A}_{r_{j}\rightarrow D}^{\text{RDR}}, \forall j, \label{eq:sss-relay-capacity}\\
\sum_{i=1}^{|\mathcal{S}|}x_{s_{i}\rightarrow r_{j}}^{\text{SRR}} &\leq& B_{r_{j}}, \forall j, \label{eq:sss-relay-flow-const}\\
\sum_{j=1}^{|\mathcal{R}|}x_{s_{i}\rightarrow r_{j}}^{\text{SRR}} &\leq& B_{s_{i}}, \forall i, \label{eq:sss-source-flow-const}\\
\underline{a}_{s_{i}} &\leq& \sum_{j=1}^{|\mathcal{R}|} a_{s_{i}\rightarrow r_{j}}^{\text{SRR}}x_{s_{i}\rightarrow r_{j}}^{\text{SRR}}, \forall i,\label{eq:sss-a-add}\\
a_{s_{i}\rightarrow r_{j}}^{\text{SRR}} &\leq& \mathcal{A}_{s_{i}\rightarrow r_{j}}^{\text{SRR}}, \forall i, \forall j,\label{eq:sss-a-leq} \\
x_{s_{i}\rightarrow r_{j}}^{\text{SRR}} &\in& \{0,1\}, \forall i, \forall j,\label{eq:sss-x-01}
\end{eqnarray}
where $a_{s_{i}\rightarrow r_{j}}^{\text{SRR}}$ and $x_{s_{i}\rightarrow r_{j}}^{\text{SRR}}$ stand for the Data rate between $s_{i}$ and $r_{j}$ and Boolean connectivity index between $s_{i}$ and $r_{j}$, respectively. Note that $s_{i}$ and $r_{j}$ stand for the source $i$, $\forall i\in \{1,\cdots,|\mathcal{S}|\}$ and the relay $j$, $\forall j\in \{1,\cdots,|\mathcal{R}|\}$, respectively.
If $s_{i}$ and $r_{j}$ are connected, $x_{s_{i}\rightarrow r_{j}}^{\text{SRR}}$ is $1$ by (\ref{eq:sss-x-01}). Otherwise, $x_{s_{i}\rightarrow r_{j}}^{\text{SRR}}$ is $0$ by (\ref{eq:sss-x-01}).
The relays are all connected to $D$, thus, $x_{r_{j}\rightarrow D}^{\text{RDR}} = 1$.
The $\mathcal{A}_{s_{i}\rightarrow r_{j}}^{\text{SRR}}$ and $\mathcal{A}_{r_{j}\rightarrow D}^{\text{RDR}}$ are maximum achievable rates computed by (\ref{eq:shannon}).
In addition, the desired data rates between $s_{i}$ and $r_{j}$ are less than $\mathcal{A}_{s_{i}\rightarrow r_{j}}^{\text{SRR}}$ as shown in (\ref{eq:sss-a-leq}) where $f_{q}\left(\cdot\right)$ is a function for the relationship between video quality and data rate (logarithmically and monotonically increasing form).
As shown in (\ref{eq:sss-a-add}), the desired data rates between $s_{i}$ and $r_{j}$ should exceed the defined minimum rates ($\underline{a}_{s_{i}}$, $\forall s_{i}$), which are minimum data rates for guaranteeing the required minimum video qualities for each flow.
Here, $\mathcal{A}_{s_{i}\rightarrow r_{j}}^{\text{SRR}}$ from $s_{i}$ to $r_{j}$ and $\mathcal{A}_{r_{j}\rightarrow D}^{\text{RDR}}$ from $r_{j}$ to $D$ are fixed.
For each individual source, there are multiple outgoing flows (multiple beams) toward relays, as formulated in (\ref{eq:sss-source-flow-const}) where $B_{s_{i}}$ stands for the number of antenna-beams at source $i, \forall i\in\{1,\cdots,|\mathcal{S}|\}$.
Similarly, each relay can form multiple beams in receiving mode, thus the number of incoming flows from sources can be $B_{r_{j}}$ as formulated in (\ref{eq:sss-relay-flow-const}) where it means the number of antenna-beams at relay $j, \forall j\in\{1,\cdots,|\mathcal{R}|\}$.
In (\ref{eq:sss-relay-capacity}), for each relay, the summation of incoming rates from sources cannot exceed the data rate between the relay and $D$.
Finally, (\ref{eq:sss-objfunc}) describes the objective of finding the pairs between sources and relays as well as finding the corresponding data rates for maximizing the total video quality and the data rate value becomes $1/2$ due to the half-duplex constraint.

\newtheorem{thm}{Theorem}
\begin{thm}
The formulation in Section~\ref{sec:jscr} is non-convex.
\end{thm}
\begin{proof}
This proof considers the setting of one-source and one-relay. Then the objective function becomes

\begin{eqnarray}
f\left(a_{s_{i}\rightarrow r_{j}}^{\text{SRR}}, x_{s_{i}\rightarrow r_{j}}^{\text{SRR}}\right)
&\triangleq&f_{q}\left(a_{s_{i}\rightarrow r_{j}}^{\text{SRR}}\right)x_{s_{i}\rightarrow r_{j}}^{\text{SRR}}\\
&=&\mathcal{K}\log_{\beta}\left(a_{s_{i}\rightarrow r_{j}}^{\text{SRR}}+1\right)x_{s_{i}\rightarrow r_{j}}^{\text{SRR}}
\end{eqnarray}
where $\mathcal{K} = \frac{1}{\log_{\beta}(a_{\text{max}}+1)}$ is constant and $x_{s_{i}\rightarrow r_{j}}^{\text{SRR}}$ is relaxed, i.e., $0\leq x_{s_{i}\rightarrow r_{j}}^{\text{SRR}}\leq 1$.
To show that this is non-convex, the second-order Hessian of this should not be positive definite~\cite{boydbook}.
The Hessian $\nabla^{2}f\left(a_{s_{i}\rightarrow r_{j}}^{\text{SRR}}, x_{s_{i}\rightarrow r_{j}}^{\text{SRR}}\right)$ is:
\begin{equation}
\begin{bmatrix}
0 & \frac{\mathcal{K}/\ln \beta}{a_{s_{i}\rightarrow r_{j}}^{\text{SRR}}+1}\\
\frac{\mathcal{K}/\ln \beta}{a_{s_{i}\rightarrow r_{j}}^{\text{SRR}}+1} & -x_{s_{i}\rightarrow r_{j}}^{\text{SRR}}\cdot\frac{\mathcal{K}/\ln \beta}{\left(a_{s_{i}\rightarrow r_{j}}^{\text{SRR}}+1\right)^{2}}
\end{bmatrix}
\end{equation}
and then the corresponding two eigenvalues are
\begin{equation}
\frac{\mathcal{I}}{2} \pm \frac{1}{2}\sqrt{
		\mathcal{I}^{2} + \left(\frac{2\mathcal{K}/\ln \beta}{a_{s_{i}\rightarrow r_{j}}^{\text{SRR}}+1}\right)^{2}
		}
\end{equation}
where $\mathcal{I}=\frac{-\frac{\mathcal{K}}{\ln \beta}\cdot x_{s_{i}\rightarrow r_{j}}^{\text{SRR}}}{\left(a_{s_{i}\rightarrow r_{j}}^{\text{SRR}}+1\right)^{2}}$, $0\leq a_{s_{i}\rightarrow r_{j}}^{\text{SRR}}\leq 1.5$, $0 \leq x_{s_{i}\rightarrow r_{j}}^{\text{SRR}} \leq 1$.

These are not all positive, thus Hessian is not positive definite, which proves the formulation is non-convex.
\end{proof}

For non-convex MINLP, heuristic searches can find approximate solutions but cannot guarantee optimality~\cite{boydbook}.
With the following Theorem, our non-convex MINLP can be re-formulated as a convex program form.

\begin{thm}
For the given formulation, (\ref{eq:sss-objfunc})-(\ref{eq:sss-x-01}), introducing
\begin{equation}
a_{s_{i}\rightarrow r_{j}}^{\text{SRR}} \leq \mathcal{A}_{s_{i}\rightarrow r_{j}}^{\text{SRR}}\cdot x_{s_{i}\rightarrow r_{j}}^{\text{SRR}}, \forall i, \forall j\label{eq:sas-milp-beauty}
\end{equation}
instead of (\ref{eq:sss-a-leq}) makes the formulation convex.
\end{thm}
\begin{proof}
For the non-convex MINLP formulation in Section~\ref{sec:jscr}, $x_{s_{i}\rightarrow r_{j}}^{\text{SRR}}=0$ means the link is disconnected. Thus the corresponding rate becomes $0$ and (\ref{eq:sas-milp-beauty}) leads to the same result when $x_{s_{i}\rightarrow r_{j}}^{\text{SRR}}=0$, i.e.,
\begin{equation}
a_{s_{i}\rightarrow r_{j}}^{\text{SRR}} \leq \mathcal{A}_{s_{i}\rightarrow r_{j}}^{\text{SRR}}\cdot 0 = 0, \forall i, \forall j.
\end{equation}
Otherwise, if $x_{s_{i}\rightarrow r_{j}}^{\text{SRR}}=1$, then this term is equivalent to (\ref{eq:sss-a-leq}).
Therefore, in turn, (\ref{eq:sss-objfunc}) is also updated as
\begin{equation}
\max \sum_{j=1}^{|\mathcal{R}|}\sum_{i=1}^{|\mathcal{S}|}f_{q}\left(\frac{1}{2}a_{s_{i}\rightarrow r_{j}}^{\text{SRR}}\right)\label{eq:objfunc-sas-simple-milp}
\end{equation}
and (\ref{eq:sss-relay-capacity}) and (\ref{eq:sss-a-add}) are also updated as following (\ref{eq:sas-relay-capacity-simple-milp}) and (\ref{eq:sss-a-add-up}):
\begin{eqnarray}
\sum_{i=1}^{|\mathcal{S}|}a_{s_{i}\rightarrow r_{j}}^{\text{SRR}} &\leq& \mathcal{A}_{r_{j}\rightarrow D}^{\text{RDR}}, \forall j.\label{eq:sas-relay-capacity-simple-milp}\\
\underline{a}_{s_{i}} &\leq& \sum_{j=1}^{|\mathcal{R}|} a_{s_{i}\rightarrow r_{j}}^{\text{SRR}}, \forall i,\label{eq:sss-a-add-up}
\end{eqnarray}
Then now there are no non-convex terms in the program.
\end{proof}

Finally, our convex MINLP, which can guarantee optimal solutions, is as follows: (\ref{eq:objfunc-sas-simple-milp}) subject to
(\ref{eq:sas-relay-capacity-simple-milp}),
(\ref{eq:sss-relay-flow-const}),
(\ref{eq:sss-source-flow-const}),
(\ref{eq:sas-milp-beauty}),
(\ref{eq:sss-a-add-up}),
(\ref{eq:sss-x-01})
where $\forall i\in\{1,\cdots,|\mathcal{S}|\}, \forall j \in \{1,\cdots,|\mathcal{R}|\}$.

\section{Performance Evaluation}\label{sec:perfeval}

To verify the performance of our scheme, i.e., \textit{video quality maximization} (named as \textsf{VQM}), we compare it with the following schemes:
\begin{itemize}
\item The joint video coding and relaying under the consideration of {\it sum rate maximization} (named as \textsf{SRM}). In this case, the proposed objective function (\ref{eq:objfunc-sas-simple-milp}) is:
\begin{equation}
\max \sum_{j=1}^{|\mathcal{R}|}\sum_{i=1}^{|\mathcal{S}|}\frac{1}{2}a_{s_{i}\rightarrow r_{j}}^{\text{SRC}}
\end{equation}
due to the fact that the quality is no longer considered.
\item The scheme in~\cite{jsac12thou}, which is an efficient algorithm that considers joint rate selection and routing (named as \textsf{JRSR}) in terms of sum-rate maximization with cooperative communication mode selection and no multiple-beams.
For fair comparisons, we adapt the scheme to our outdoor-stadium architecture (one-tier relay) and allow only decode-and-forward relaying.
\end{itemize}

\begin{table}[tp]%
\caption{Expectation of Achieved Normalized Aggregated Video Quality}
\label{expectation}
\centering %
\begin{tabular}{r|r|r||r|r|r}
\hline
\multicolumn{3}{c||}{}&\multicolumn{3}{c}{Multiple-Beams at $s_{i}$ and $r_{j}$} \\
\hline
$|\mathcal{S}|$ & $|\mathcal{R}|$ & Setting & \textsf{VQM} & \textsf{SRM} & \textsf{JRSR}\\
\hline
\hline
5 & 10 &   I & 4.166 & 3.873 & 3.331 \\
5 & 10 &  II & 4.934 & 4.647 & 4.165 \\
5 & 10 & III & 4.681 & 4.397 & 3.632 \\
\hline
10 & 5 &   I & 4.164 & 3.871 & 3.352 \\
10 & 5 &  II & 4.954 & 4.620 & 4.182 \\
10 & 5 & III & 4.664 & 4.371 & 3.650 \\
\hline
10 & 10 &   I & 8.813 & 8.451 & 5.483 \\
10 & 10 &  II & 9.817 & 9.452 & 6.336 \\
10 & 10 & III & 9.312 & 8.958 & 5.795 \\
\hline
10 & 15 &   I & 8.883 & 8.574 & 5.633 \\
10 & 15 &  II & 9.872 & 9.563 & 6.456 \\
10 & 15 & III & 9.383 & 9.074 & 5.927 \\
\hline
15 & 10 &   I & 13.420 & 11.765 & 6.483 \\
15 & 10 &  II & 14.902 & 13.255 & 7.376 \\
15 & 10 & III & 14.403 & 12.755 & 6.839 \\
\hline
\end{tabular}
\end{table}

For the setting, the cameras are uniformly distributed on top of the stadium. Between stadium and broadcasting center, multiple relays are uniformly deployed along a line. To vary the settings, we consider this line to be near the cameras (Setting I), in the middle between cameras and broadcasting center (Setting II), and near the center (Setting III).
As our performance measure, we consider the cumulative probability distribution (cdf) of the aggregate video quality.
The cdf is obtained as follows: we consider multiple realizations of the deployment of sources and relays, i.e., the random deployment of relays with Setting I, Setting II, and Setting III.
For each such realization, we optimize coding rates and relay selection; thus each run gives us one realization of the aggregate video quality. We finally plot the cdf of this quality. For the simulation of \textsf{VQM}, the lower bounds ($\underline{a}_{s_{i}}$, $\forall i\in\{1,\cdots,|\mathcal{S}|\}$) of each source are set as $0.75$ Gbit/s ($50\%$ of $1.5$ Gbit/s).
More detailed scenarios and simulation results can be found in~\cite{kim12tbc}.

\subsection{CDF of Aggregate Video Quality}\label{sec:sim01}

\begin{figure}[t!]
	\begin{center}
		\includegraphics[scale=0.5]{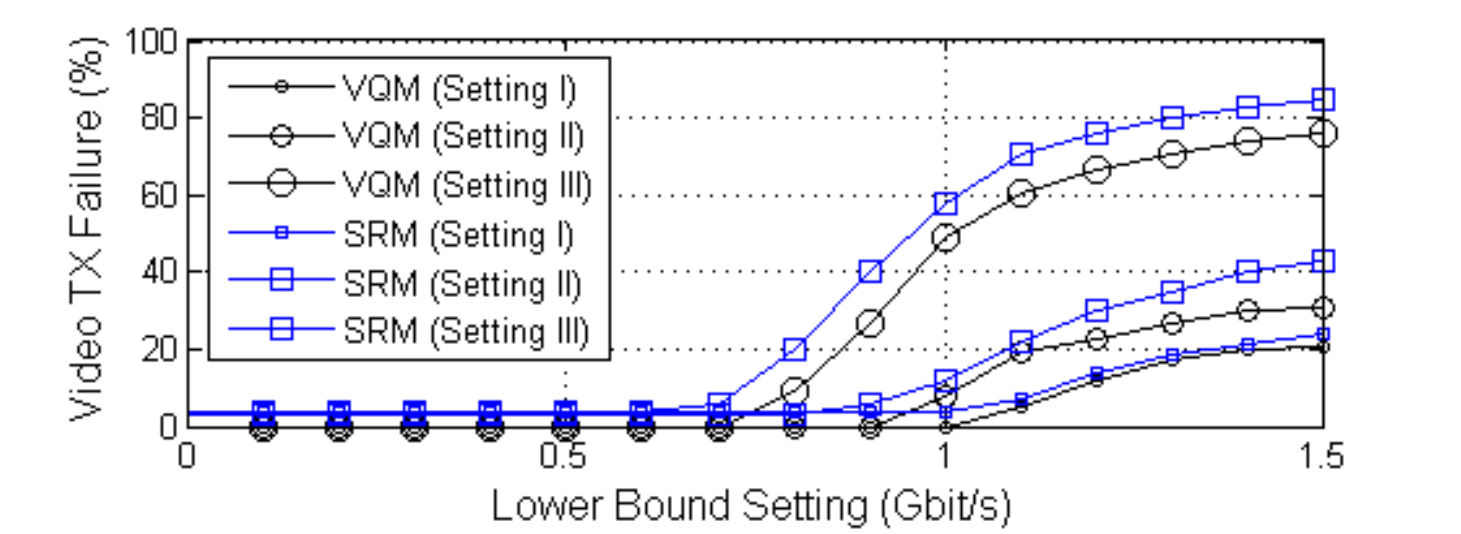}
	\end{center}
	\caption{Impact of Various Lower Bound Setting: $|\mathcal{S}|=10, |\mathcal{R}|=15$}
	\label{fig:lower}
\end{figure}	

\begin{figure*}[t!]
	\centering
	\subfigure[Setting I]{
		\includegraphics[scale=0.38]{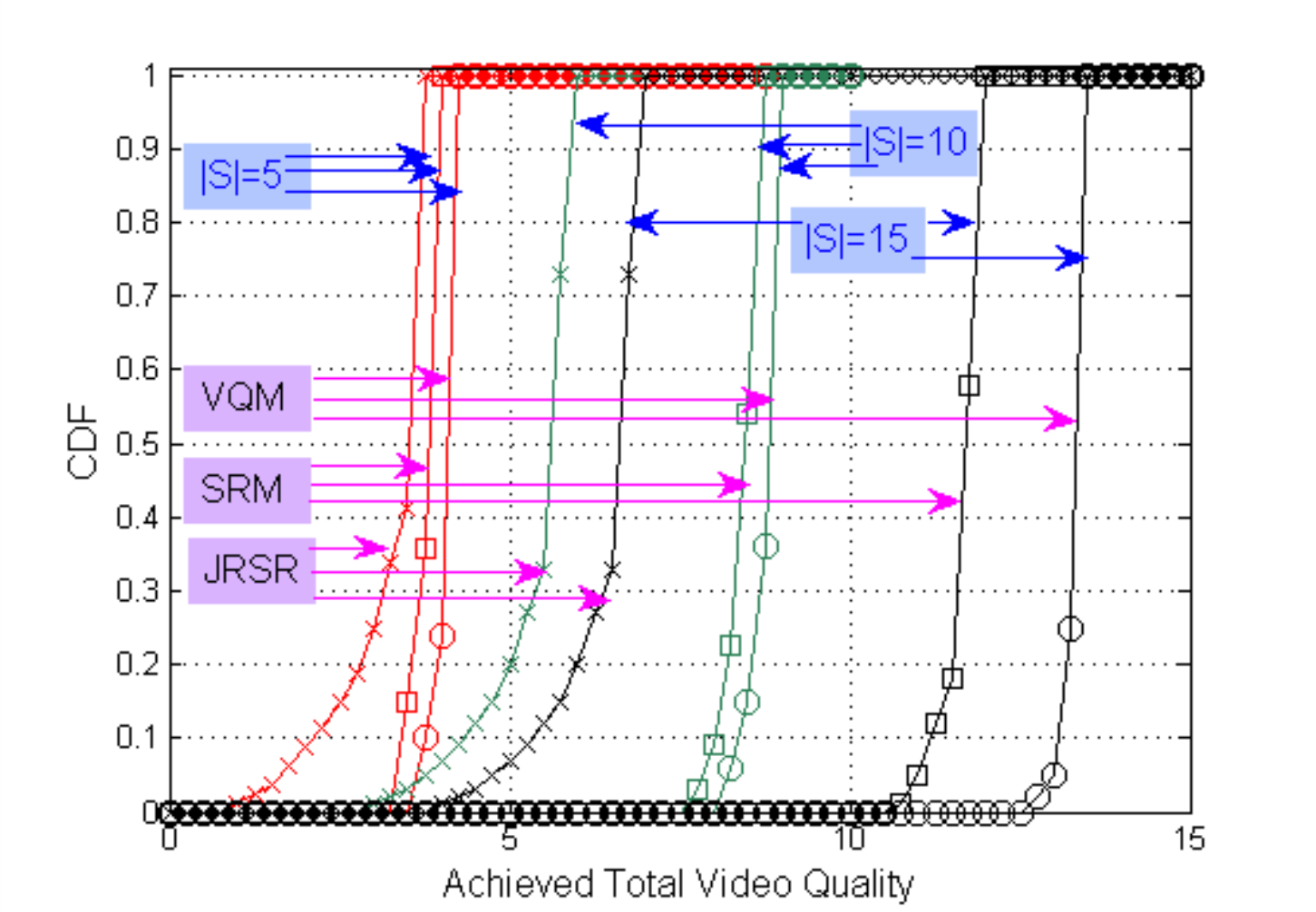}
		\label{fig:cdf_msmr_r_a}
	}
	\subfigure[Setting II]{
		\includegraphics[scale=0.38]{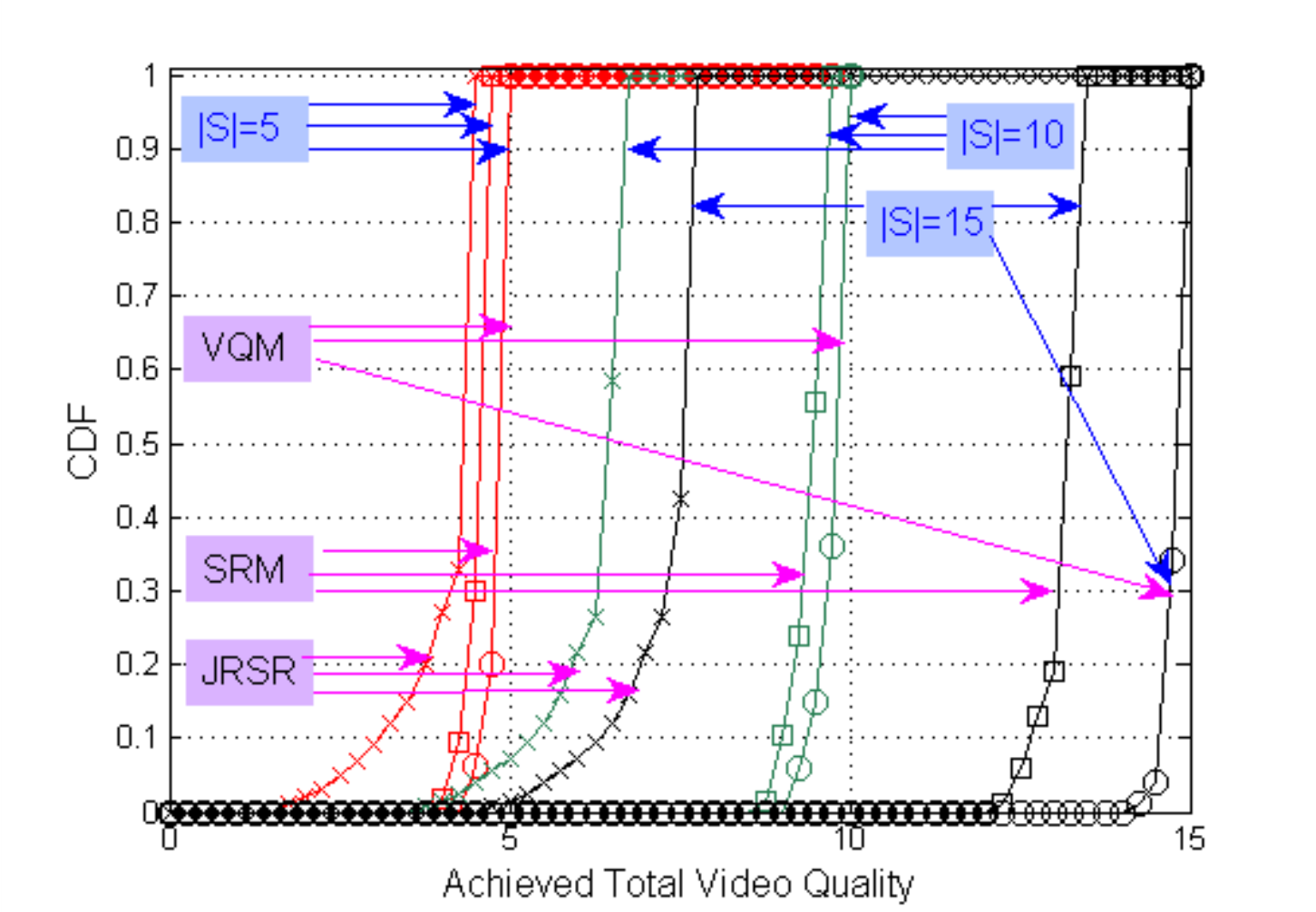}
		\label{fig:cdf_msmr_r_b}
	}
	\subfigure[Setting III]{
		\includegraphics[scale=0.38]{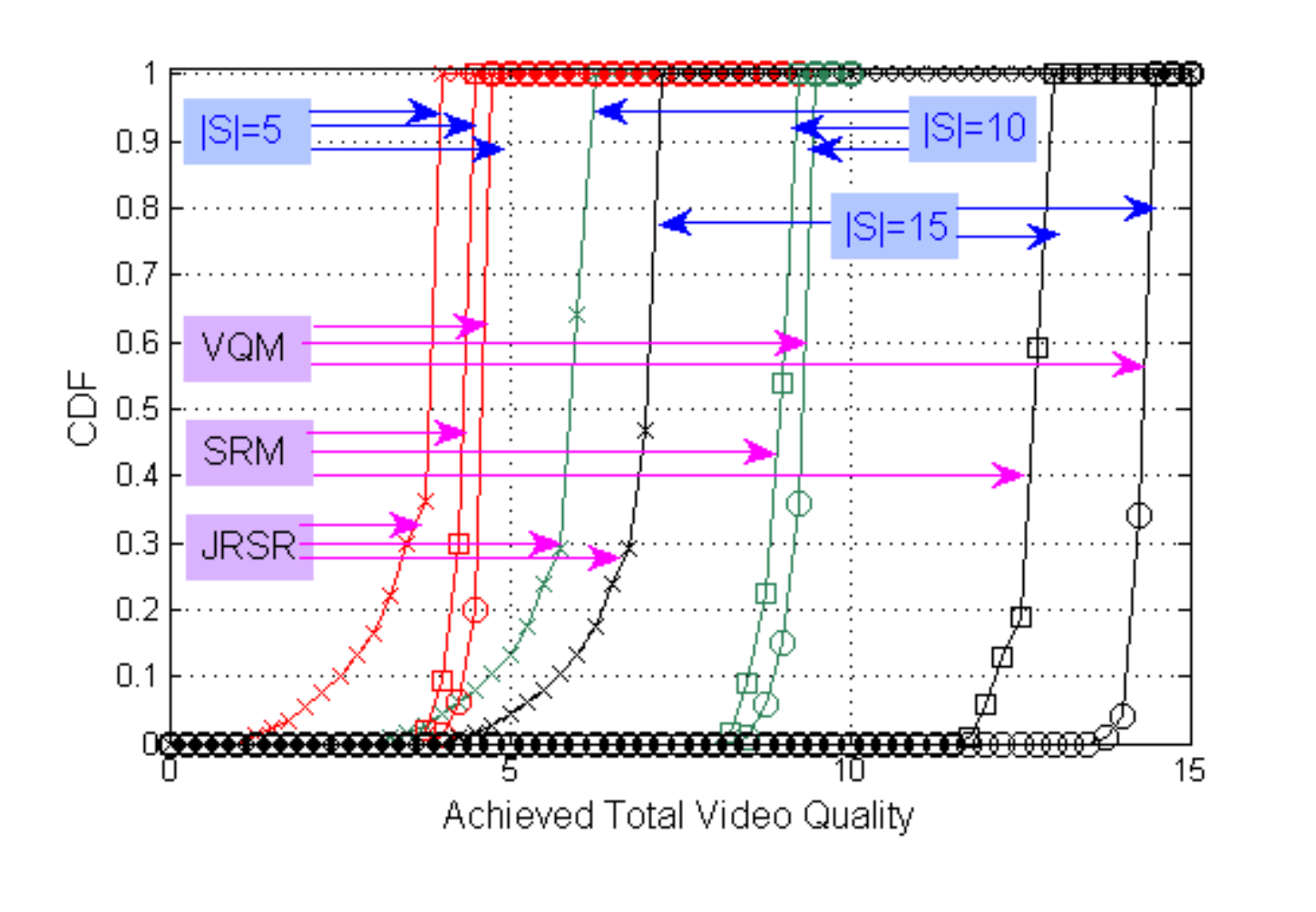}
		\label{fig:cdf_msmr_r_c}
	}
	\caption{Simulation Results: Number of Sources ($|\mathcal{S}|=5, 10, 15$) and Fixed Number of Relays ($|\mathcal{R}|=10$)}
	\label{fig:cdf_msmr_r}
\end{figure*}

Fig.~\ref{fig:cdf_msmr_r} plots the cases that the number of sources is smaller, equal, or larger than the number of relays (i.e., $|\mathcal{S}|=5, 10, 15$, and $|\mathcal{R}|=10$).
The mean achieved normalized aggregated video qualities are in Table~\ref{expectation}.
In Table~\ref{expectation}, the video quality values from each source are normalized as $1$ for performance evaluation.
Thus, if we have $N$ cameras in the system, the maximum achievable aggregated video quality is $N$.
As shown in this result, the performance of \textsf{JRSR} is worse than that of both \textsf{SRM} and \textsf{VQM}.
The latter (i.e., \textsf{JRSR}), by design, does not allow the exploitation of the multiple-beam antennas at relays, and thus shows worse performance.
The performance gains of \textsf{SRM} are more pronounced than \textsf{JRSR} in Settings I and III, i.e., when the relays are close to either the sources or the destination. More importantly, we find that the relative performance advantage drastically increases as the number of sources increases relative to the number of relays. This is not surprising, as for these situations the ability of the sources to split their streams and flexibly route them via the relays becomes more important.
We also see that \textsf{SRM} shows lower performance than \textsf{VQM} due to the fact that \textsf{SRM} aims to the maximization of sum data rates, while \textsf{VQM} aims to maximize the overall delivered video quality.
The relative advantage of \textsf{VQM} also increases as the number of sources increases. Again, this is not surprising, as the bandwidth limitations become more stringent as the number of sources increases.


\subsection{Impact of Lower Bound Setting}\label{sec:sim04}
In previous simulation, the lower bounds for the data rate per data stream are set as $0.75$ Gbit/s. Here, we vary this value from $0$ Gbit/s (no lower bound) to $1.5$ Gbit/s (allowing only uncompressed video) in steps of $0.1$ Gbit/s. As a performance quality measure, we define ``stream outage" (i.e., the probability that at least one stream does not have the minimum required quality).
As shown in Fig.~\ref{fig:lower}, Setting III suffers significantly from the higher required per-stream quality.
With Setting III, the data rates between sources and relays are lower than the others.
Thus, when we set the lower bound quite high, all flows are disconnected. Thus, it achieves the lowest performance.
On the other hand, in Setting I, all flows between sources and relays have enough capacity to support uncompressed video transmission, thus, a higher setting for minimum quality does not have a strong impact.
Fig.~\ref{fig:lower} also shows that \textsf{VQM} has better performance than \textsf{SRM} for all settings.
\section{Conclusion}\label{sec:conclusions}
This paper suggests and discusses quality-aware coding and routing for 60\,GHz multi-Gbit/s real-time video streaming in an outdoor broadcasting system.
In the system, there are multiple wireless video cameras distributed throughout the stadium.
We presented an optimization framework for finding the combination of wireless link pairs between wireless cameras and relays that can maximize the overall or per-flow qualities of delivered video to a broadcasting center.
An initial non-convex MINLP is re-formulated as a convex program, which allows optimum solutions. Simulations show that this methodology outperforms other methods that do not take the peculiarities of millimeter-wave video links into account.


\end{document}